\documentclass[11pt]{article}
\usepackage{geometry}
\usepackage{graphicx}
\usepackage{amsmath, amssymb, amsfonts, amsthm}
\usepackage{authblk}
\usepackage{algorithm}
\usepackage{algpseudocode}

\usepackage{xspace}
\usepackage{cite}

\usepackage{enumerate}
\usepackage{hyperref}
\usepackage{paralist}
\usepackage{gensymb}
\usepackage{rotating}
\usepackage{multirow}
\usepackage{verbatim} 
\usepackage{mleftright}
\usepackage{todonotes}

\usepackage{amsfonts,latexsym}
\usepackage{fullpage,color}
\usepackage{url}

\newtheorem{lemma}{Lemma}

\newtheorem{theorem}[lemma]{Theorem}
\newtheorem{definition}[lemma]{Definition}

\begin{document}

\title{Roundtrip Spanners with $(2k-1)$ Stretch}

\author[1]{Ruoxu Cen\thanks{cenbo@aliyun.com}}
\author[1]{Ran Duan\thanks{duanran@mail.tsinghua.edu.cn}}
\author[1]{Yong Gu\thanks{guyong12@mails.tsinghua.edu.cn}}

\affil[1]{Institute for Interdisciplinary Information Sciences, Tsinghua University}

\maketitle

\begin{abstract}  
A roundtrip spanner of a directed graph $G$ is a subgraph of $G$ preserving roundtrip distances approximately for all pairs of vertices. Despite extensive research, there is still a small stretch gap between roundtrip spanners in directed graphs and undirected graphs. For a directed graph with real edge weights in $[1,W]$, we first propose a new deterministic algorithm that constructs a roundtrip spanner with $(2k-1)$ stretch and $O(k n^{1+1/k}\log (nW))$ edges for every integer $k> 1$, then remove the dependence of size on $W$ to give a roundtrip spanner with $(2k-1)$ stretch and $O(k n^{1+1/k}\log n)$ edges. While keeping the edge size small, our result improves the previous $2k+\epsilon$ stretch roundtrip spanners in directed graphs [Roditty, Thorup, Zwick'02; Zhu, Lam'18], and almost matches the undirected $(2k-1)$-spanner with $O(n^{1+1/k})$ edges [Alth\"ofer et al. '93] when $k$ is a constant, which is optimal under Erd\"os conjecture.
\end{abstract}

\clearpage

\section{Introduction}

A $t$-spanner of a graph $G$ is a subgraph of $G$ in which the distance between every pair of vertices is at most $t$ times their distance in $G$, where $t$ is called the stretch of the spanner.
Sparse spanner is an important choice to implicitly representing all-pair distances \cite{Zwick2001}, and
spanners also have application backgrounds in distributed systems (see~\cite{Peleg2000}).
For undirected graphs, $(2k-1)$-spanner with $O(n^{1+1/k})$ edges is proposed and conjectured to be optimal \cite{Althofer1993, Thorup2001}.
However, directed graphs may not have sparse spanners with respect to the normal distance measure.
For instance, in a bipartite graph with two sides $U$ and $V$, if there is a directed edge from every vertex in $U$ to every vertex in $V$,
then removing any edge $(u,v)$ in this graph will destroy the reachability from $u$ to $v$, so its only spanner is itself, which has $O(n^2)$ edges.
To circumvent this obstacle, one can approximate the optimal spanner in terms of edge size (e.g. in~\cite{Dinitz2011,Berman2011}), or one can define directed spanners on different distance measures.
This paper will study directed sparse spanners on roundtrip distances.

Roundtrip distance is a natural metric with good property.
Cowen and Wagner \cite{Cowen1999,Cowen2000} first introduce it into directed spanners.
Formally, roundtrip distance between vertices $u,v$ in $G$ is defined as $d_G(u\leftrightarrows v)=d_G(u\to v)+d_G(v\to u)$, where $d_G(u\to v)$ is the length of shortest path from $u$ to $v$ in $G$.
For a directed graph $G=(V,E)$, a subgraph $G'=(V,E')$ ($E'\subseteq E$) is called a $t$-roundtrip spanner of $G$ if for all $u,v\in G$, $d_{G'}(u\leftrightarrows v)\le t\cdot d_G(u\leftrightarrows v)$,
where $t$ is called the stretch of the roundtrip spanner.

In a directed graph $G=(V,E)$ ($n=|V|,m=|E|$) with real edge weights in $[1,W]$,
Roditty et al.\ \cite{Roditty:2008} give a $(2k+\epsilon)$-spanner of $O(\min\{(k^2/\epsilon) n^{1+1/k}\log(nW)$, $(k/\epsilon)^2 n^{1+1/k}(\log n)^{2-1/k}\})$ edges.
Recently, Zhu and Lam \cite{Zhu2018} derandomize it and improve the size of the spanner to $O((k/\epsilon) n^{1+1/k}\log(nW))$ edges, while the stretch is also $2k+\epsilon$.
We make a step further based on these works and reduce the stretch to $2k-1$. Formally, we state our main results in the following theorems.

\begin{theorem}
\label{the:main1}
For any directed graph $G$ with real edge weights in $[1,W]$ and integer $k\ge 1$,
there exists a $(2k-1)$-roundtrip spanner of $G$ with $O(k n^{1+1/k}\log (nW))$ edges, which can be constructed in $\tilde{O}(kmn\log W)$ time\footnote{$\tilde{O}(\cdot)$ hides $\log n$ factors.}.
\end{theorem}

By a similar scaling method in~\cite{Roditty:2008}, we can make the size of the spanner independent of the maximum edge weight $W$ to obtain a $(2k-1)$-spanner with strongly subquadratic space.

\begin{theorem}
\label{the:main2}
For any directed graph $G$ with real edge weights in $[1,W]$ and integer $k\ge 1$,
there exists a $(2k-1)$-roundtrip spanner of $G$ with $O(k n^{1+1/k}\log n)$ edges, which can be constructed in $\tilde{O}(kmn\log W)$ time.
\end{theorem}

Actually, our result almost matches the lower bound following girth conjecture.
The girth conjecture, implicitly mentioned by Erd\"os \cite{Erdos1964}, says that for any $k$, there exists a graph with $n$ vertices and $\Omega(n^{1+1/k})$ edges whose girth (minimum cycle) is at least $2k+2$.
This conjecture implies that no algorithm can construct a spanner of $O(n^{1+1/k})$ size and less than $2k-1$ stretch for all undirected graph with $n$ vertices \cite{Thorup2001}.
This lower bound also holds for roundtrip spanners on directed graphs.

Our approach is based on the scaling constructions of the $(2k+\epsilon)$-stretch roundtrip spanners in~\cite{Roditty:2008, Zhu2018}. To reduce the stretch, we construct inward and outward shortest path trees from vertices in a hitting set~\cite{Aingworth1999,Dor2000} of size $O(n^{1/k})$, and carefully choose the order to process vertices in order to make the stretch exactly $2k-1$. To further make the size of the spanner strongly subquatratic, we use a similar approach as in~\cite{Roditty:2008} to contract small edges in every scale, and treat vertices with different radii of balls of size $n^{1-1/k}$ differently.

\subsection{Related Works}

The construction time in this paper is $\tilde{O}(kmn\log W)$. However, there exist roundtrip spanners with $o(mn)$ construction time but larger stretches.
Pachoci et al.\ \cite{Pachocki2018} proposes an algorithm which can construct $O(k\log n)$-roundtrip spanner with $O(n^{1+1/k}\log^2 n)$ edges. Its construction time is $O(mn^{1/k}\log^5 n)$, which breaks the cubic time barrier.
Very recently, Chechik et al.\ \cite{Chechik2019} give an algorithm which constructs $O(k\log\log n)$-roundtrip spanners with $\tilde{O}(n^{1+1/k})$ edges in $\tilde{O}(m^{1+1/k})$ time.

For spanners defined with respect to normal directed distance, researchers aim to approximate the $k$-spanner with minimum number of edges.
Dinitz and Krauthgamer \cite{Dinitz2011} achieve $\tilde{O}(n^{2/3})$ approximation in terms of edge size, and Bermen et al. \cite{Berman2011} improves the approximation ratio to $\tilde{O}(n^{1/2})$.

Another type of directed spanners is transitive-closure spanner, introduced by Bhattacharyya et al.\ \cite{Bhattacharyya2012}. In this setting the answer may not be a subgraph of $G$, but a subgraph of the transitive closure of $G$. In other words, selecting edges outside the graph is permitted. The tradeoff is between diameter (maximum distance) and edge size. One of Bhattacharyya et al.'s results is spanners with diameter $k$ and $O((n\log n)^{1-1/k})$ approximation of optimal edge size~\cite{Bhattacharyya2012}, using a combination of linear programming rounding and sampling. Berman et al. \cite{Berman2010} improves the approximation ratio to $O(n^{1-1/[k/2]}\log n)$. We refer to Raskhodnikova \cite{Raskhodnikova2010} as a review of transitive-closure spanners.

\subsection{Organization}
In Section~\ref{sec:pre}, the notations and basic concepts used in this paper will be discussed. In Section~\ref{sec:algo} we describe the construction of the $(2k-1)$-roundtrip spanner with $O(kn^{1+1/k}\log(nW))$ edges, thus proving Theorem~\ref{the:main1}. Then in Section~\ref{sec:strongly} we improve the size of the spanner to $O(kn^{1+1/k}\log n)$ and still keep the stretch to $(2k-1)$, thus proving Theorem~\ref{the:main2}. The conclusion and further direction are discussed in Section~\ref{sec:conclusion}.

\section{Preliminaries}\label{sec:pre}

In this paper we consider a directed graph $G=(V,E)$ with non-negative real edge weights $w$ where $w(e)\in [1,W]$ for all $e\in E$.
Denote $G[U]$ to be the subgraph of $G$ induced by $U\subseteq V$, i.e.\ $G[U]=(U,E\cap(U\times U))$.
A roundtrip path between nodes $u$ and $v$ is a cycle (not necessarily simple) passing through $u$ and $v$.
The roundtrip distance between $u$ and $v$ is the minimum length of roundtrip paths between $u$ and $v$.
Denote $d_U(u\leftrightarrows v)$ to be the roundtrip distance between $u$ and $v$ in $G[U]$. (Sometimes we may also use $d_U(u\leftrightarrows v)$ to denote a roundtrip shortest path between $u,v$ in $G[U]$.) It satisfies:
\begin{itemize}
    \item For $u,v\in U$, $d_U(u\leftrightarrows u)=0$ and $d_U(u\leftrightarrows v)=d_U(v\leftrightarrows u)$.
    \item For $u,v\in U$, $d_U(u\leftrightarrows v)=d_U(u\to v)+d_U(v\to u)$.
    \item For $u,v,w\in U$, $d_U(u\leftrightarrows v)\leq d_U(u\leftrightarrows w)+d_U(w\leftrightarrows v)$.
\end{itemize}
Here $d_U(u\to v)$ is the one-way distance from $u$ to $v$ in $G[U]$. We use $d(u\leftrightarrows v)$ to denote the roundtrip distance between $u$ and $v$ in the original graph $G=(V,E)$.

In $G$, a $t$-roundtrip spanner of $G$ is a subgraph $H$ of $G$ on the same vertex set $V$ such that the roundtrip distance between any pair of $u,v\in V$ in $H$ is at most $t\cdot d(u\leftrightarrows v)$. $t$ is called the \emph{stretch} of the spanner.

For a subset of vertices $U\subseteq V$, given a center $u\in U$ and a radius $R$, define roundtrip ball $Ball_U(u,R)$ to be the set of vertices whose roundtrip distance on $G[U]$ to center $u$ is strictly smaller than the radius $R$. Formally, $Ball_U(u,R)=\{v\in U:d_U(u\leftrightarrows v)< R\}$. Then the size of the ball, denoted by $|Ball_U(u,R)|$, is the number of vertices in it. Similarly we define $\overline{Ball}_U(u,R)=\{v\in U:d_U(u\leftrightarrows v)\le R\}$.
Subroutine \verb;InOutTrees;$(U,u,R)$ calculates the edge set of an inward and an outward shortest path tree centered at $u$ spanning vertices in $Ball_U(u,R)$ on $G[U]$. (That is, the shortest path tree from $u$ to all vertices in $Ball_U(u,R)$ and the shortest path tree from all vertices in $Ball_U(u,R)$ to $u$.)
It is easy to see that the shortest path trees will not contain vertices outside $Ball_U(u,R)$:

\begin{lemma}
The inward and outward shortest path trees returned by \verb;InOutTrees;$(U,u,R)$ only contain vertices in $Ball_U(u,R)$.
\end{lemma}
\begin{proof}
For any $v\in Ball_U(u,R)$, let $C$ be a cycle containing $u$ and $v$ such that the length of $C$ is less than $R$. Then for any vertex $w\in C$, $d_U(u\leftrightarrows w)<R$, so $w$ must be also in the trees returned by \verb;InOutTrees;$(U,u,R)$.
\end{proof}

For all notations above, we can omit the subscript $V$ when the roundtrip distance is considered in the original graph $G=(V,E)$. Our algorithm relies on the following well-known theorem to calculate hitting sets deterministically.

\begin{theorem}{(Cf.\ Aingworth et al.\ \cite{Aingworth1999}, Dor et al.\ \cite{Dor2000})}
\label{the:hitset}
For universe $V$ and its subsets $S_1,S_2,\ldots, S_n$, if $|V|=n$ and the size of each $S_i$ is greater than $p$, then there exists a hitting set $H\subseteq V$ intersecting all $S_i$, whose size $|H|\le (n\ln n)/p$, and such a set $H$ can be found in $O(np)$ time deterministically.
\end{theorem}

\section{A $(2k-1)$-Roundtrip Spanner Algorithm}
\label{sec:algo}
In this section we introduce our main algorithm constructing a $(2k-1)$-roundtrip spanner with $O(k n^{1+1/k}\log(nW))$ edges for any $G$.
We may assume $k\ge 2$ in the following analysis, since the result is trivial for $k=1$.

Our approach combines the ideas of~\cite{Roditty:2008} and~\cite{Zhu2018}. In~\cite{Zhu2018}, given a length $L$, we pick an arbitrary vertex $u$ and find the smallest integer $h$ such that $|\overline{Ball}(u,(h+1)L)|<n^{1/k}|\overline{Ball}(u,h\cdot L)|$, then we include the inward and outward shortest path tree centered at $u$ spanning $\overline{Ball}(u,(h+1)L)$ and remove vertices in $\overline{Ball}(u,h\cdot L)$ from $V$. We can see that $h\leq k$, so the stretch is $2k$ for $u,v$ with roundtrip distance $L$, and by a scaling approach the final stretch is $2k+\epsilon$. We observe that if $h=k-1$, $|\overline{Ball}(u,(k-1)L)|\geq n^{(k-1)/k}$, so by Theorem~\ref{the:hitset} we can preprocess the graph by choosing a hitting set $H$ with size $O(n^{1/k}\log n)$ and construct inward and outward shortest path trees centered at all vertices in $H$, then we do not need to include the shortest path trees spanning $\overline{Ball}(u,k\cdot L)$. The stretch can then be decreased to $2k-1+\epsilon$. To make the stretch equal $2k-1$, instead of arbitrarily selecting $u$ each time, we carefully define the order to select $u$.

\subsection{Preprocessing}
\label{sec:preprocess}
We first define a radius $R(u)$ for each vertex $u$. It is crucial for the processing order of vertices.

\begin{definition}
For all $u\in V$, we define $R(u)$ to be the maximum length $R$ such that $|Ball(u,R)|< n^{1-1/k}$, that is, if we sort the vertices by their roundtrip distance to $u$ in $G$ by increasing order, $R(u)$ is the roundtrip distance from $u$ to the $\lceil n^{1-1/k}\rceil$-th vertex.
\end{definition}

For any $u\in V$, $|\overline{Ball}(u,R(u))|\ge n^{1-1/k}$.
By Theorem \ref{the:hitset}, we can find a hitting set $H$ intersecting all sets in $\{\overline{Ball}(u,R(u)):u\in V\}$, such that $|H|=O(n^{1/k}\log n)$.
For all $t\in H$, we build an inward and an outward shortest path tree of $G$ centered at $t$, and denote the set of edges of these trees by $E_0$ and include them in the final spanner.
This step generates $O(n^{1+1/k}\log n)$ edges in total, and it is easy to obtain the following statement: 
\begin{lemma}\label{lem:large-distance}
For $u,v\in V$ such that $d(u\leftrightarrows v)\geq R(u)/(k-1)$, the roundtrip distance between $u$ and $v$ in the graph $(V,E_0)$ is at most $(2k-1)d(u\leftrightarrows v)$.
\end{lemma}
\begin{proof}
Find the vertex $t\in H$ such that $t\in \overline{Ball}(u,R(u))$, that is, $d(u\leftrightarrows t)\leq R(u)$. Then the inward and outward shortest path trees from $t$ will include $d(u\leftrightarrows t)$ and $d(t\leftrightarrows v)$. By $R(u)\leq (k-1)d(u\leftrightarrows v)$, we have $d(u\leftrightarrows t)\leq (k-1)d(u\leftrightarrows v)$ and $d(t\leftrightarrows v)\leq d(t\leftrightarrows u)+d(u\leftrightarrows v)\leq k\cdot d(u\leftrightarrows v)$.
So the roundtrip distance of $u$ and $v$ in $E_0$ is at most $d(u\leftrightarrows t)+d(t\leftrightarrows v)\leq (2k-1)d(u\leftrightarrows v)$.
\end{proof}

\subsection{Approximating a Length Interval}
Instead of approximating all roundtrip distances at once, we start with an easier subproblem of approximating all pairs of vertices whose roundtrip distances are within an interval $[L/(1+\epsilon), L)$.
Parameter $\epsilon$ is a real number in $(0,1/(2k-2)]$. The procedure \verb;Cover;($G,k,L,\epsilon$) described in Algorithm~\ref{alg:cover} will return a set of edges which gives a $(2k-2)(1+\epsilon)$-approximation of roundtrip distance $d(u\leftrightarrows v)$ if $R(u)/(k-1)>d(u\leftrightarrows v)$, for $d(u\leftrightarrows v)\in [L/(1+\epsilon), L)$.

\begin{algorithm}
\caption{Cover($G(V,E), k, L, \epsilon$)\label{alg:cover}}
\begin{algorithmic}[1]
\State $U\gets V, \hat{E}=\emptyset$
\While{$U\neq \emptyset$}
\State $u \gets \arg\max_{u\in U}R(u)$
\State $step \gets \min\{R(u)/(k-1), L\}$
\State $h\gets$ minimum positive integer satisfying $|Ball_U(u,h\cdot step)|< n^{h/k}$ \label{lin:select-h}
\State Add \Call{InOutTrees}{$U,u,h\cdot step$} to $\hat{E}$
\State Remove $\overline{Ball}_U(u, (h-1)step)$ from $U$
\EndWhile
\State \Return $\hat{E}$
\end{algorithmic}
\end{algorithm}

Note that in Algorithm~\ref{alg:cover}, initially $U=V$ and the balls are considered in $G[U]=G$. In the end of every iteration we remove a ball from $U$, and the following balls are based on the roundtrip distances in $G[U]$. However, $R(u)$ does not need to change during the algorithm and can still be based on roundtrip distances in the original graph $G$. The analysis for the size of the returned set $\hat{E}$ and the stretch are as follows.
\begin{lemma}
\label{lem:size}
The returned edge set of \verb;Cover;($G,k,L,\epsilon$) has $O(n^{1+1/k})$ size.
\end{lemma}
\begin{proof}
When processing a vertex $u$, by the selection of $h$ in line~\ref{lin:select-h}, $|Ball_U(u,h\cdot step)|< n^{h/k}$ and $|\overline{Ball}_U(u,(h-1)step)|\ge n^{(h-1)/k}$. When $h\ge 2$ it is because of $h$'s minimality, and when $h=1$ it is because $u\in \overline{Ball}_U(u,0)$.
So each time \verb;InOutTrees; is called, the size of ball to build shortest path trees is no more than $n^{1/k}$ times the size of ball to remove.
During an execution of \verb;Cover;($G,k,L,\epsilon$), each vertex is removed
once from $U$.
Therefore the total number of edges added in $\hat{E}$ is $O(n^{1+1/k})$.
\end{proof}

We can also see that if the procedure \verb;Cover;($G[U],k,L,\epsilon$) is run on a subgraph $G[U]$ induced on a subset $U\subseteq V$, then the size of $\hat{E}$ is bounded by $O(|U|n^{1/k})$. It is also easy to see that $h$ is at most $k-1$:
\begin{lemma}\label{lem:h}
The $h$ selected at line~\ref{lin:select-h} in \verb;Cover;($G,k,L,\epsilon$) satisfies $h\leq k-1$.
\end{lemma}
\begin{proof}
In $G[U]$, the ball $Ball_U(u,(k-1)step)$ must have size no greater than $Ball(u,(k-1)step)$ since the distances in $G[U]$ cannot decrease while some vertices are removed. Since $|Ball(u,R(u))|<n^{1-1/k}$ and $step\leq R(u)/(k-1)$, we get $|Ball_U(u,(k-1)step)|\leq |Ball(u,(k-1)step)|<n^{1-1/k}$, thus $h\leq k-1$.
\end{proof}

Next we analyze the roundtrip distance stretch in $\hat{E}$. Note that in order to make the final stretch $2k-1$, for the roundtrip distance approximated by edges in $\hat{E}$ we can make the stretch $(2k-2)(1+\epsilon)$, but for the roundtrip distance approximated by $E_0$ we need to make the stretch at most $2k-1$ as $E_0$ stays the same.
\begin{lemma}
\label{lem:stretch}
For any pair of vertices $u,v$ such that $d(u\leftrightarrows v)\in[L/(1+\epsilon), L)$, either \verb;Cover;($G,k,L,\epsilon$)'s returned edge set $\hat{E}$ can form a cycle passing through $u,v$ with length at most $(2k-2)(1+\epsilon)d(u\leftrightarrows v)$,
or $R(u)\leq (k-1)d(u\leftrightarrows v)$, in which case the $E_0$ built in Section~\ref{sec:preprocess} can form a detour cycle with length at most $(2k-1)d(u\leftrightarrows v)$ by Lemma~\ref{lem:large-distance}.
\end{lemma}
\begin{proof}
Consider any pair of vertices $u,v$ with roundtrip distance $d=d(u\leftrightarrows v)\in[L/(1+\epsilon), L)$, and a shortest cycle $P$ going through $u,v$ with length $d$.

During \verb;Cover;($G,k,L,\epsilon$), consider the vertices on $P$ that are first removed from $U$. Suppose $w$ is one of the first removed vertices, and $w$ is removed as a member of $\overline{Ball}_{U_c}(c,(h_c-1)step_c)$ centered at $c$.
This is to say $d_{U_c}(c\leftrightarrows w)\leq (h_c-1)step_c$.

Case 1: $step_c> d$. Then \[d_{U_c}(c\leftrightarrows u)\le d_{U_c}(c\leftrightarrows w)+d_{U_c}(w\leftrightarrows u)\leq (h_c-1)step_c+d< h_c step_c,\] and $u\in Ball_{U_c}(c, h_c step_c)$.
The second inequality holds because $U_c$ is the remaining vertex set before removing $w$, so by definition of $w$, all vertices on $P$ are in $U_c$.
Symmetrically $v\in Ball_{U_c}(c,h_c step_c)$.
\verb;InOutTrees;($U_c,c,h_c step_c)$ builds a detour cycle passing through $u,v$ with length $<2h_c step_c$.
By Lemma~\ref{lem:h}, we have $h_c\leq k-1$. Also $step_c\leq L\leq (1+\epsilon)d$, therefore we build a detour of length $<2(k-1)step_c\le (2k-2)(1+\epsilon)d$ in $\hat{E}$.

Case 2: $step_c\leq d$.
Because $d<L$, this case can only occur when $step_c=R(c)/(k-1)$.
Because $c$ is chosen before $u$, $R(u)\le R(c)=(k-1)step_c\leq (k-1)d$.
By Lemma~\ref{lem:large-distance}, $E_0$ can give a $(2k-1)$-approximation of $d$.


\end{proof}

\subsection{Main Construction}
\label{sec:mainW}
Now we can proceed to prove 
the main theorem based on a scaling on lengths of the cycles from $1$ to $2nW$.

\begin{theorem}
\label{the:mainW}
For any directed graph $G$ with real edge weights in $[1,W]$,
there exists a polynomial time constructible $(2k-1)$-roundtrip spanner of $G$ with $O(k n^{1+1/k}\log(nW))$ edges.
\end{theorem}

\begin{proof}
Note that the roundtrip distance between any pair of vertices must be in the range $[1,2(n-1)W]$.
First do the preprocessing in Section \ref{sec:preprocess}.
Then divide the range of roundtrip distance $[1,2nW)$ into intervals $[(1+\epsilon)^{p-1}, (1+\epsilon)^p)$, where $\epsilon=1/(2k-2)$.
Call \verb;Cover;($G, k, (1+\epsilon)^p, \epsilon$) for $p=0,\cdots,\lfloor \log_{1+\epsilon}(2nW) \rfloor+1$, and merge all returned edges with $E_0$ to form a spanner.

First we prove that the edge size is $O(kn^{1+1/k}\log(nW))$.
Preprocessing adds $O(n^{1+1/k}\cdot \log n)$ edges.
\verb;Cover;($G, k, (1+\epsilon)^p, \epsilon$) is called for $\log_{1+1/(2k-2)}(2nW)=O(k\log(nW))$
times.
By Lemma~\ref{lem:size}, each call generates $O(n^{1+1/k})$ edges.
So the total number of edges in the roundtrip spanner is $O(kn^{1+1/k}\log(nW))$.

Next we prove the stretch is $2k-1$.
For any pair of vertices $u,v$ with roundtrip distance $d$, let $p=\lfloor \log_{1+\epsilon} d \rfloor+1$, then $d\in[(1+\epsilon)^{p-1}, (1+\epsilon)^p)$.
By Lemma \ref{lem:stretch}, either the returned edge set of \verb;Cover;($G,k,(1+\epsilon)^p, \epsilon$) can form a detour cycle passing through $u,v$ of length at most $(2k-2)(1+\epsilon)d=(2k-1)d$, or the edges in $E_0$ can form a detour cycle passing through $u,v$ of length at most $(2k-1)d$.

In conclusion this algorithm can construct a $(2k-1)$-roundtrip spanner with $O(kn^{1+1/k}\cdot\log(nW))$ edges.
\end{proof}

\subsection{Construction Time}
\label{sec:timeW}
The running time of the algorithm in the proof of Theorem~\ref{the:mainW} is $O(kn(m+n\log n)\log(nW))$. It is also easy to see that the algorithm is deterministic.
Next we analyze construction time in detail.

In preprocessing, for any $u\in V$, $R(u)$ can be calculated by running Dijkstra searches with Fibonacci heap \cite{Fredman1987} starting at $u$, so calculating $R(\cdot)$ takes $O(n(m+n\log n))$ time.
Finding $H$ takes $O(n^{2-1/k})$ time by Theorem~\ref{the:hitset}. Building $E_0$ takes $O(n^{1/k}\log n\cdot(m+n\log n))$ time.

A \verb;Cover; call's while loop runs at most $n$ times since each time at least one node is removed.
In a loop, $u$ can be found in $O(n)$ time, and all other operations regarding roundtrip balls can be done in $O(m+n\log n)$ time by Dijkstra searches starting at $u$ on $G[U]$.
Therefore a \verb;Cover; call takes $O(n(m+n\log n))$ time.

\verb;Cover; is called $O(k\log(nW))$ times.
Combined with the preprocessing time, the total construction time is $O(kn(m+n\log n)\log(nW))$.

\section{Removing the Dependence on $W$}\label{sec:strongly}
In this section we prove Theorem~\ref{the:main2}. The size of the roundtrip spanner in Section~\ref{sec:algo} is dependent on the maximum edge weight $W$. In this section we remove the dependence by designing the scaling approach more carefully. Our idea is similar to that in~\cite{Roditty:2008}. When we consider the roundtrip distances in the range $[L/(1+\epsilon), L)$, all cycles with length $\leq L/n^3$ have little effect so we can contract them into one node, and all edges with length $>(2k-1)L$ cannot be in any $(2k-1)L$ detour cycles, so they can be deleted. Thus, an edge with length $l$ can only be in  $O(\log_{1+\epsilon} n)$ iterations for $L$ between $l/(2k-1)$ and $l\cdot n^3$ (based on the girth of this edge). However, the stretch will be a little longer if we directly apply the algorithm in Section~\ref{sec:algo} on the contracted graph.

To overcome this obstacle, we only apply the vertex contraction when $R(u)$ is large (larger than $2(k-1)L$). By making the ``step'' a little larger than $L$ and $\epsilon$ smaller, when $d < L<step$, the stretch is still bounded by $(2k-1)$. When $R(u)\leq 2(k-1)L$, we first delete all node $v$ with $R(v) < L/8$, then simply apply the algorithm in Section~\ref{sec:algo} in the original graph. Since every node $u$ can only be in the second part when $R(u)/2(k-1) \leq L\leq 8R(u)$, the number of edges added in the second part is also strongly polynomial.

First we define the girth of an edge:

\begin{definition}
We define the \emph{girth of an edge $e$} in $G$ to be the length of shortest directed cycle containing $e$, and denote it by $g(e)$.
\end{definition}

It is easy to see that for $e=(u,v)$, $d(u\leftrightarrows v)\leq g(e)$. In $O(n(m+n\log n))$ time we can compute $g(e)$ for all edges $e$ in $G$~\cite{Fredman1987}. 

Algorithm~\ref{alg:1} approximates roundtrip distance $d(u\leftrightarrows v)\in[L/(1+\epsilon),L)$. 
In the $p$-th iteration of the algorithm, $G_p[U_p]$ is always the subgraph contracted from the subgraph $G[U]$. Given $v_p\in U_p$, let $C(v_p)$ be the set of vertices in $U$ that are contracted into $v_p$. We can see the second part of this algorithm (after line~\ref{line:rm-small-R}) is the same as Algorithm~\ref{alg:cover} in Section~\ref{sec:algo}.

For the contracted subgraph $G_p[U_p]$, we give new definitions for balls and \verb;InOutTrees;.
Given two vertices $u_p,v_p\in U_p$, define 
\[\hat{d}_{U_p}(u_p,v_p)=\min_{u\in C(u_p),v\in C(v_p)} d_U(u,v)\]
Balls in $G_p[U_p]$ are defined as follows.
    \[Ball_{U_p}(u_p,r)=\{v_p\in U_p : \hat{d}_{U_p}(u_p,v_p) < r\}\]
    \[\overline{Ball}_{U_p}(u_p,r)=\{v_p\in U_p : \hat{d}_{U_p}(u_p,v_p)\le r\}\]
In Line~\ref{line:inouttrees}, \verb;NewInOutTrees;$(U_p,u_p,h\cdot step)$ is formed by only keeping the edges between different contracted vertices in \verb;InOutTrees;$(U,u,h\cdot step)$ from $u$ (see Line~\ref{line:select_u}). In the inward tree or outward tree of \verb;InOutTrees;$(U,u,h\cdot step)$, if after contraction there are multiple edges from or to a contracted vertex, respectively, only keep one of them. We can see the number of edges added to $\hat{E}$ is bounded by $O(|Ball_{U_p}(u_p,h\cdot step)|)$. Also in $U_p$, the roundtrip distance from $u_p$ to vertices in $Ball_{U_p}(u_p,h\cdot step)$ by edges in \verb;NewInOutTrees;$(U_p,u_p,h\cdot step)$ is at most $h\cdot step$.

In line~\ref{line:delete-long}, we can delete long edges since obviously they cannot be included in $\hat{E}$.

\begin{algorithm}[tb]
	\caption{Cover2$(G,k,p,\epsilon)$}\label{alg:1}
	\begin{algorithmic}[1]
	    \State{$L\gets (1+\epsilon)^p$}
	    \State {Contract all edges $e$ with $g(e)\le L/n^3$ in $G$ to form a graph $G_p$, let its vertex set be $V_p$}
	    \State {(Delete edges $e$ with $g(e)>2(k-1)L$ from $G_p$)} \label{line:delete-long}
		\State {$U\gets V, U_p\gets V_p, \hat{E}\gets\emptyset$}
		\While {$U\neq\emptyset$ and $\max_{u\in U} R(u)\geq 2(k-1)L$}
		    \State {$u\gets \arg\max_{u\in U} R(u)$, let $u_p$ be the corresponding vertex in $U_p$}\label{line:select_u}
			\State {$step\gets (1+1/n^2)L$}\label{line:step}
			\State {$h\gets$ minimum positive integer satisfying $|Ball_{U_p}(u_p,h\cdot step)|<n^{h/k}$}\label{line:ball}
			\State {Add \Call{NewInOutTrees}{$U_p,u_p,h\cdot step$} to $\hat{E}$}\label{line:inouttrees}
			\State {Remove $\overline{Ball}_{U_p}(u_p,(h-1)step)$ from $U_p$, remove corresponding vertices from $U$}\label{line:remove}
		\EndWhile
		\State {Remove all vertices $u$ from $U$ with $R(u) < L/8$}\label{line:rm-small-R}
		\While {$U\neq\emptyset$}
		    \State {$u\gets \arg\max_{u\in U} R(u)$}
		    \State {$step\gets \min\{R(u)/(k-1),L\}$}\label{line:step2}
			\State {$h\gets$ minimum positive integer satisfying $|Ball_U(u,h\cdot step)|<n^{h/k}$}\label{line:ball2}
			\State {Add \Call{InOutTrees}{$U,u,h\cdot step$} to $\hat{E}$}\label{line:inouttrees2}
			\State {Remove $\overline{Ball}_{U}(u,(h-1)step)$ from $U$}\label{line:remove2}
		\EndWhile
		\State \Return {$\hat{E}$}
	\end{algorithmic}
\end{algorithm}

The main algorithm is shown in Algorithm~\ref{alg:2}.

\begin{algorithm}[ht]
    \caption{Spanner$(G(V,E),k)$}\label{alg:2}
    \begin{algorithmic}[1]
        \State {Do the preprocessing in Section~\ref{sec:preprocess}. Let $E_0$ be the added edges}
        \State {$\epsilon\gets \frac{1}{4(k-1)}$.}
        \State {$\hat{E}\gets E_0$}
        \For {$p\gets 0$ to $\lfloor\log_{1+\epsilon}(2nW)\rfloor+1$}
            
            \State{$\hat{E}\gets\hat{E}\cup$ \Call{Cover2}{$G,k,p,\epsilon$}}
        \EndFor
        \State \Return {$H(V,\hat{E})$}
    \end{algorithmic}
\end{algorithm}

\begin{lemma}
    For $k\leq n$ and $n\geq 12$, Algorithm \emph{Spanner}$(G,k)$ constructs a $(2k-1)$-roundtrip spanner of $G$.
\end{lemma}
\begin{proof}
    For any pair of vertices $u,v$ with roundtrip distance $d=d(u\leftrightarrows v)$ on $G$, there exists a $p$, such that $d\in[(1+\epsilon)^{p-1}, (1+\epsilon)^p)$. Let $L=(1+\epsilon)^p$. If $R(u) \le (k-1)d$ or $R(v) \le (k-1)d$, by Lemma~\ref{lem:large-distance}, $E_0$ contains a roundtrip cycle between $u$ and $v$ with length at most $(2k-1)d$. So we assume $R(u) > (k-1)d$ and $R(v)>(k-1)d$. Also, if there is a vertex $w$ on the shortest cycle containing $u$ and $v$ with $R(w)<L/8$, then there will be a vertex $t\in H$ so that $d(w\leftrightarrows t)<L/8$, so the roundtrip distance in $E_0$ will be $d(u\leftrightarrows t)+d(t\leftrightarrows v)<L/4+2d\leq (1+\epsilon)d/4+2d<(2k-1)d$ for $k\geq 2$, so Line~\ref{line:rm-small-R} cannot impact the correctness.
    
    Consider the iteration $p$ of Algorithm~\ref{alg:1}, let $u_p,v_p$ be the contracted vertices of $u,v$ respectively. Let $P$ be a shortest cycle going through $u,v$ in $G$ and $P'$ be the contracted cycle going through $u_p,v_p$ in $G_p$. It is easy to see that each vertex on $P$ corresponds to some vertex on $P'$. Similar as Lemma~\ref{lem:h}, in Line~\ref{line:ball} and Line~\ref{line:ball2}, we have $(k-1)\cdot step \le R(u)$. It is easy to see that $|Ball_{U_p}(u_p,(k-1)\cdot step)|\le |Ball_U(u,(k-1)\cdot step)| < n^{1-1/k}$, which implies $h\le k-1$.
    
    We prove it by the induction on $p$. When $p$ is small, there is no contracted vertex in $G_p$. By the same argument as in Lemma~\ref{lem:stretch}, either \verb;Cover2;$(G,k,L,\epsilon)$'s returned edge set $\hat{E}$ contains a roundtrip cycle between $u$ and $v$ with length at most \[2h_c\cdot step_c\le 2(k-1)(1+1/n^2)(1+\epsilon)d = (2k-3/2)(1+1/n^2)d \le (2k-1)d\] ($k\ge 2$, $k\leq n$ and $n\geq 12$) since $step_c\le (1+1/n^2)L$ in Line~\ref{line:step} and Line~\ref{line:step2} and $h_c\le k-1$, or $E_0$ contains a cycle between $u$ and $v$ with length at most $(2k-1)d$.
    Next we assume that vertices of $G$ contracted in the same vertex in $G_p$ are already connected in $\hat{E}$, and has the $(2k-1)$-stretch.
    
    During \verb;Cover2;$(G,k,p,\epsilon)$, if some vertices in $P'$ are removed from $U_p$ in Line~\ref{line:remove}, like Lemma~\ref{lem:stretch}, suppose $w_p\in U_p$ is one of the first removed vertices, and $w_p$ is removed as a member of $\overline{Ball}_{U_c}(c,(h_c-1)step_c)$ centered at $c$. Let $w'\in C(w_p)$ be one vertex on $P$, since there are at most $n$ original vertices contracted and $step_c = (1+1/n^2)L$, we have $d_U(c\leftrightarrows u)\le d_{U_p}(c\leftrightarrows w_p) + d_U(w'\leftrightarrows u) + n\cdot L/n^3 \le (h_c-1)step_c + d_U(u\leftrightarrows v) + L/n^2 < (h_c-1)step_c + L + L/n^2 = h_c step_c$, and symmetrically $d_U(c\leftrightarrows v)<h_c step_c$. Thus \verb;NewInOutTrees;$(U_c,c,h_c step_c)$ builds a roundtrip cycle passing through $u_p,v_p$ of length $<2h_c step_c$ in current contracted graph. It follows that $d_{G_p[\hat{E}]}(u_p,v_p)< 2h_c step_c\le  2(k-1)(1+1/n^2)L$. Since there are at most $n$ contracted vertices in the roundtrip cycle between $u_p$ and $v_p$, and $w(e)\leq g(e)$ for every contracted edge $e$, we have \[d_{G[\hat{E}]}(u,v) \le 2(k-1)(1+1/n^2)L + n\cdot (2k-1) L/n^3\le (2k-3/2)(1+3/n^2)d \le (2k-1)d.\] ($k\ge 2$, $k\leq n$ and $n\geq 12$.)
    
    If there is no vertex in $P'$ removed from $U_p$ in Line~\ref{line:remove} and Line~\ref{line:rm-small-R}, then all vertices $w$ in $P$ have $L/8\leq R(w)<2(k-1)L$. By the same argument as in Lemma~\ref{lem:stretch}, the second part of Algorithm~\ref{alg:1} also ensures that $\hat{E}\cup E_0$ contains a roundtrip cycle passing through $u,v$ with length at most $(2k-1)d$.
\end{proof}

\begin{lemma}
    The subgraph returned by algorithm \emph{Spanner}$(G,k)$ has $O(kn^{1+1/k}\log n)$ edges.
\end{lemma}
\begin{proof}
    Preprocessing adds $O(n^{1+1/k}\log n)$ edges as in Section~\ref{sec:preprocess}. The edges added in Line~\ref{line:inouttrees2} is bounded as follows. Consider Algorithm~\ref{alg:1}, after Line~\ref{line:rm-small-R}, the subgraph only consists of vertices with $R(u)\in [L/8, 2(k-1)L]$, so each vertex belongs to at most $\log_{1+\epsilon} 16k$ such iterations. Thus the total number of edges added after Line~\ref{line:rm-small-R} is at most $n^{1+1/k}\log_{1+\epsilon} 16k = O(kn^{1+1/k}\log k)$ edges. Next we count the edges added in Line~\ref{line:inouttrees}.
    
    We remove the directions of all edges in $G$ to get an undirected graph $G'$, and remove the directions of all edges in every $G_p$ to get an undirected graph $G_p'$, but define the weight of an edge $e$ in $G'$ and every $G_p'$ to be the girth $g(e)$ in $G$. Let $F$ be a minimum spanning forest of $G'$ w.r.t. the girth $g(e)$. We can see that in iteration $p$, if we remove edges in $F$ with $g(e)>2(k-1)(1+\epsilon)^p$ and contract edges $e$ with $g(e)\leq (1+\epsilon)^p/n^3$ in $F$, then the connected components in $F$ will just be the connected components in $G_p'$, which are the strongly connected components in $G_p$. This is because of the cycle property of MST: If an edge $e=(u,v)$ in $G_p'$ has $g(e)\leq (1+\epsilon)^p/n^3$, then in $F$ all edges $f$ on the path connecting $u,v$ have $g(f)\leq (1+\epsilon)^p/n^3$, thus $u,v$ are already contracted in $F$; If an edge $e=(u,v)$ in $G_p'$ has $g(e)\leq 2(k-1)(1+\epsilon)^p$, then in $F$ all edges $f$ on the path connecting $u,v$ have $g(f)\leq 2(k-1)(1+\epsilon)^p$, so $u,v$ are in the same component in $F$.


    So the total size of connected components $\{C: |C|\geq 2\}$ in $G_p'$ is at most 2 times the number of edges $e$ in $F$ with $(1+\epsilon)^p/n^3<g(e)\leq 2(k-1)(1+\epsilon)^p$, and every edge in $F$ can be in at most $\log_{1+\epsilon}2(k-1)n^3=O(k\log n)$ number of different $G_p'$. Thus, the total size of connected components with size at least 2 in all $G_p'$ is bounded by $O(kn\log n)$.
    By a similar argument of Lemma~\ref{lem:size}, in each call of \verb;Cover2;($G,k,p,\epsilon$), line~\ref{line:inouttrees} will add $|C|n^{1/k}$ new edges to $\hat{E}$, for every connected component $C$ with $|C|\geq 2$ in $G_p'$. Thus the total number of edges in the subgraph returned by \verb;Spanner;($G,k$) is bounded by $O(kn^{1+1/k}\log n)$.
\end{proof}

\subsection{Construction Time}
The analysis of \verb;Spanner;'s running time is similar to Section~\ref{sec:timeW}.
Compared with \verb;Cover;, \verb;Cover2; adds operations of building $G_p$.
We also need to calculate $g(\cdot)$ in preprocessing, which can done by $n$ Dijkstra searches.
$G_p$ can be built in $O(m)$ time.
\verb;Cover2; is called $\log_{1+\epsilon'}(2nW)=O(k\log(nW))$ times.
Therefore the total construction time is still $O(kn(m+n\log n)\log(nW))$.

\section{Conclusion}\label{sec:conclusion}
In this paper we discuss the construction of $(2k-1)$-roundtrip spanners with $O(kn^{1+1/k}\log n)$ edges. An important and interesting further direction is whether we can find truly subcubic algorithm constructing such spanners.

\bibliography{refs}

\end{document}